\documentclass[12pt]{iopart}
\expandafter\let\csname equation*\endcsname\relax
\expandafter\let\csname endequation*\endcsname\relax

\usepackage[utf8]{inputenc}
\usepackage[english]{babel}
\usepackage[T1]{fontenc}
\usepackage{amsmath}
\usepackage{hyperref}

\usepackage{graphics}
\usepackage{graphicx}
\usepackage{amsthm}
\usepackage{here}
\usepackage{arydshln}
\usepackage{setspace}
\usepackage{wrapfig}
\usepackage{array, amssymb, latexsym, bm, mathtools, braket, multirow, url}
\usepackage{color}
\usepackage{cancel}
\usepackage{comment}
\usepackage{ascmac}
\usepackage{tikz}
\usepackage{paralist}

\usepackage[all]{xy}
\usepackage{xr}

\newtheorem{theorem}{Theorem}
\newtheorem{definition}[theorem]{Definition}
\newtheorem{lemma}[theorem]{Lemma}

\newtheorem{corollary}[theorem]{Corollary}

\newtheorem{remark}[theorem]{Remark}

\newcommand{\cC}{\mathcal{C}}

\newcommand{\cH}{\mathcal{H}}

\newcommand{\cL}{\mathcal{L}}
\newcommand{\cM}{\mathcal{M}}

\newcommand{\cS}{\mathcal{S}}

\newcommand{\cV}{\mathcal{V}}

\DeclareMathOperator{\rank}{rank}

\newcommand{\ketbra}[2]{\ket{#1}\!\bra{#2}}

\newcommand{\Her}[1]{\cL_{\mathrm{H}}(#1)}
\newcommand{\Psd}[1]{\cL_{\mathrm{H}}^+(#1)}

\begin{document}

\title{Derivation of Standard Quantum Theory via State Discrimination}

\author{Hayato Arai}
%\email{m18003b@math.nagoya-u.ac.jp}
\address{Graduate School of Mathematics, Nagoya University, Furo-cho, Chikusa-ku, Nagoya, 464-8602, Japan}
\address{(Current Affiliation) Mathematical Quantum Information RIKEN Hakubi Research Team, 
RIKEN Cluster for Pioneering Research (CPR) and RIKEN Center for Quantum Computing (RQC), 
Wako, Saitama 351-0198, Japan.
}
%\orcid{0000-0002-6780-1246}

\eads{\mailto{hayato.arai@riken.jp}}

\author{Masahito Hayashi}
%\email{hayashi@sustech.edu.cn}
%\email{hmasahito@cuhk.edu.cn}
\address{School of Data Science, The Chinese University of Hong Kong, Shenzhen, Longgang District, Shenzhen, 518172, China}
%\email{masahito@math.nagoya-u.ac.jp}
\address{International Quantum Academy (SIQA), Shenzhen 518048, China}
\address{Graduate School of Mathematics, Nagoya University, Furo-cho, Chikusa-ku, Nagoya, 464-8602, Japan}
%\orcid{0000-0003-3104-1000}

\eads{\mailto{masahito@math.nagoya-u.ac.jp}}

\vspace{10pt}
\begin{indented}
\item[] Fub. 2024.
\end{indented}

\begin{abstract}
It is a key issue to characterize the model of standard quantum theory out of general models by an operational condition.
The framework of General Probabilistic Theories (GPTs) is a new information theoretical approach
to single out standard quantum theory.
It is known that traditional properties, for example, Bell-CHSH inequality, are not sufficient to single out standard quantum theory among possible models in GPTs.
As a more precise property,
we focus on the bound of the performance for an information task called state discrimination in general models.
We give an equivalent condition for outperforming the minimum discrimination error probability under the standard quantum theory given by the trace norm.
Besides, by applying the equivalent condition,
we characterize standard quantum theory out of general models in GPTs by the bound of the performance for state discrimination.
\end{abstract}

\vspace{2pc}

\submitto{\NJP}

\maketitle

\section{Introduction}
The mathematical model of quantum theory described by Hilbert space expresses our physical systems very well.
However, a foundation of the mathematical model is not completely clarified,
and therefore, many researchers have been discussing a foundation of quantum theory from several viewpoints.
Recently, as research of quantum information theory has flourished,
informational theoretical viewpoints have received much attention.
A modern information theoretical approach is General Probabilistic Theories (GPTs) \cite{
PR1994,Plavala2017,Pawlowski2009,BBB2010,Banik:2012,Stevens:2013,
Barnum.Steering:2013,LNS.et.al.2022,AYH2023,Takakura2022,KMI2009,KNI2010-dist,NKM2010,
KNI2010,Muller2013,BLSS2017,Barnum2019,Arai2019,YAH2020,
Janotta2014,Plavala2021,Short2010,Barnum2012,Barrett2007,CDP2010,CS2015,CS2016,
KBBM2017,Matsumoto2018,Takagi2019,ALP2019,ALP2021,MAB2022,AH2022,NLS2022,PNL2023,RLW2022},
which start with statistics from states and measurements.
Simply speaking,
GPTs deal with all models where the law to obtain measurement outcomes is given by non-negative probability distribution.

The central purpose of studies of GPTs is to derive the model of quantum theory from information theoretical principles.
Especially, a bound of performance for some fundamental information tasks is important
because it can be regarded as an implicit limitation in our physical experimental setting.
Studies of CHSH inequality in GPTs \cite{PR1994,Plavala2017,Pawlowski2009,BBB2010,Banik:2012,Stevens:2013,
Barnum.Steering:2013,LNS.et.al.2022,AYH2023,Takakura2022} are typical trials of such a derivation from operational bounds.
However, it is known that Tsirelson's bound cannot single out the model of quantum theory \cite{BBB2010,Banik:2012,Stevens:2013,Barnum.Steering:2013,LNS.et.al.2022,AYH2023}.
In the current situation,
it is known that mathematical properties about information theoretical objects, for example, symmetry of perfectly distinguishable pure states, can single out quantum theory \cite{KNI2010,Muller2013,BLSS2017,Barnum2019}.
Such derivations are suggestive, but 
their physical meaning is not clear.
On the other hand, a derivation from 
a simple bound of information tasks has a clear meaning but is still open.

In order to single out quantum theory from a bound of information tasks,
we need to find an information task to satisfy the following requirement; 
With respect to this task,
the performance under a general model 
outperforms the performance under the standard quantum theory.
Typical known results of such fundamental information tasks are the results of perfect discrimination of non-orthogonal states \cite{Arai2019,YAH2020}.
The references \cite{Arai2019,YAH2020} show that certain classes of beyond-quantum measurements can perfectly discriminate a pair of two non-orthogonal states.
However, the analyses in the references \cite{Arai2019,YAH2020}
are insufficient to single out quantum theory for the following reasons.
First, they deal only with perfect state discrimination.
Second, they only deal with a type of measurement with a parameter.
In other words, they do not clarify a precise condition when a measurement improves the performance for state discrimination.
The above two reasons prevent us from singling out quantum theory.

This paper aims to resolve the above two weaknesses and to derive standard quantum theory through the performance for state discrimination.
Therefore, this paper deals with 
state discrimination in more general models of GPTs
without imposing perfectness,
and we give an equivalent condition 
for a measurement to outperform 
the state discrimination by measurements in the standard quantum theory.

One of the most important values in state discrimination
is
the total error probability
\begin{align}\label{def:error}
	\mathrm{Err}(\rho_0,\rho_1;p;\bm{M}):=p\Tr \rho_0M_1+(1-p)\Tr \rho_1M_0
\end{align}
with two hypotheses $\rho_0,\rho_1$
generated with probability $p$ and $1-p$
by a measurement $\bm{M}:=\{M_0,M_1\}$.
In standard quantum theory,
i.e., by applying a POVM $\bm{M}$,
a tight bound of $\mathrm{Err}(\rho_0,\rho_1;p;\bm{M})$ is given \cite{Holevo1972,HelstromBook1976,HayashiBook2017} as
\begin{align}\label{eq:trace-norm}
	\mathrm{Err}(\rho_0;\rho_1;p;\bm{M})\ge\frac{1}{2}-\frac{1}{2}\left\|p\rho_0-(1-p)\rho_1\right\|_1.
\end{align}
Hence, we seek an equivalent condition when the minimization of the error probability $\mathrm{Err}(\rho_0;\rho_1;p;\bm{M})$ in a general model is smaller
than the value $\frac{1}{2}-\frac{1}{2}\|p\rho_0-(1-p)\rho_1\|_1$.
In order to deal with the trace norm $\|\cdot\|_1$,
this paper discusses models whose state can be described as a Hermitian matrix, called a \textit{quantum-like model} at first.
This restriction looks strong, but we show that any model with a condition of its dimension essentially satisfies this requirement (Lemma~\ref{theorem:isometry}).

In terms of the comparison of the error $\mathrm{Err}(\rho_0;\rho_1;p;\bm{M})$,
perfect discrimination of non-orthogonal states in \cite{Arai2019,YAH2020} corresponds to the compatibility
of the relations
$\mathrm{Err}(\rho_0;\rho_1;p;\bm{M})=0$ and $\frac{1}{2}-\frac{1}{2}\|\rho_0-\rho_1\|_1>0$
in the case of $p=\frac{1}{2}$.
In this sense,
the references \cite{Arai2019,YAH2020} show the possibility of violating 
the bound of POVMs \eqref{eq:trace-norm}.
In this paper,
we answer a more general question, i.e., ``When does a measurement in a general model violate the bound of POVMs?''
First, we give a general bound of $\mathrm{Err}(\rho_0;\rho_1;p;\bm{M})$ with its equality condition (Theorem~\ref{theorem:performance}).
By using the general bound, we give an equivalent condition for the existence of a tuple of $\rho_0,\rho_1,\bm{M}$ satisfying $\mathrm{Err}(\rho_0;\rho_1;p;\bm{M})<\frac{1}{2}-\frac{1}{2}\|p\rho_0-(1-p)\rho_1\|_1$ in a model in the case of $p=1/2$ (Theorem~\ref{theorem:equivalent}).

Moreover, this paper answers the first purpose of GPTs,
i.e., characterization of standard quantum theory via the quantum bound for state discrimination \eqref{eq:trace-norm}.
By applying our equivalent condition,
we derive standard quantum theory by the existence embedding of state space satisfying the quantum bound \eqref{eq:trace-norm} (Theorem~\ref{theorem:derivation}).
Because the quantum bound \eqref{eq:trace-norm} derives the model of quantum theory,
the performance for state discrimination completely characterizes any other properties in standard quantum thery,
which is a surprising operational meaning of our derivation.
There exist many measures of the performance for information tasks that outperform the limit under the standard quantum theory in certain models of GPTs \cite{PR1994,Pawlowski2009,Barnum.Steering:2013,Arai2019,YAH2020,Short2010,Barnum2012,NLS2022,PNL2023}.
However, all such known results do not completely characterize the model of quantum theory,
i.e.,
such measures sometimes behave in the same way as standard quantum theory.
Our finding is that the quantum bound of the performance for state discrimination is a complete measure (Figure~\ref{figure_class}).

\begin{figure}[h]
	\centering
	\includegraphics[width=8cm]{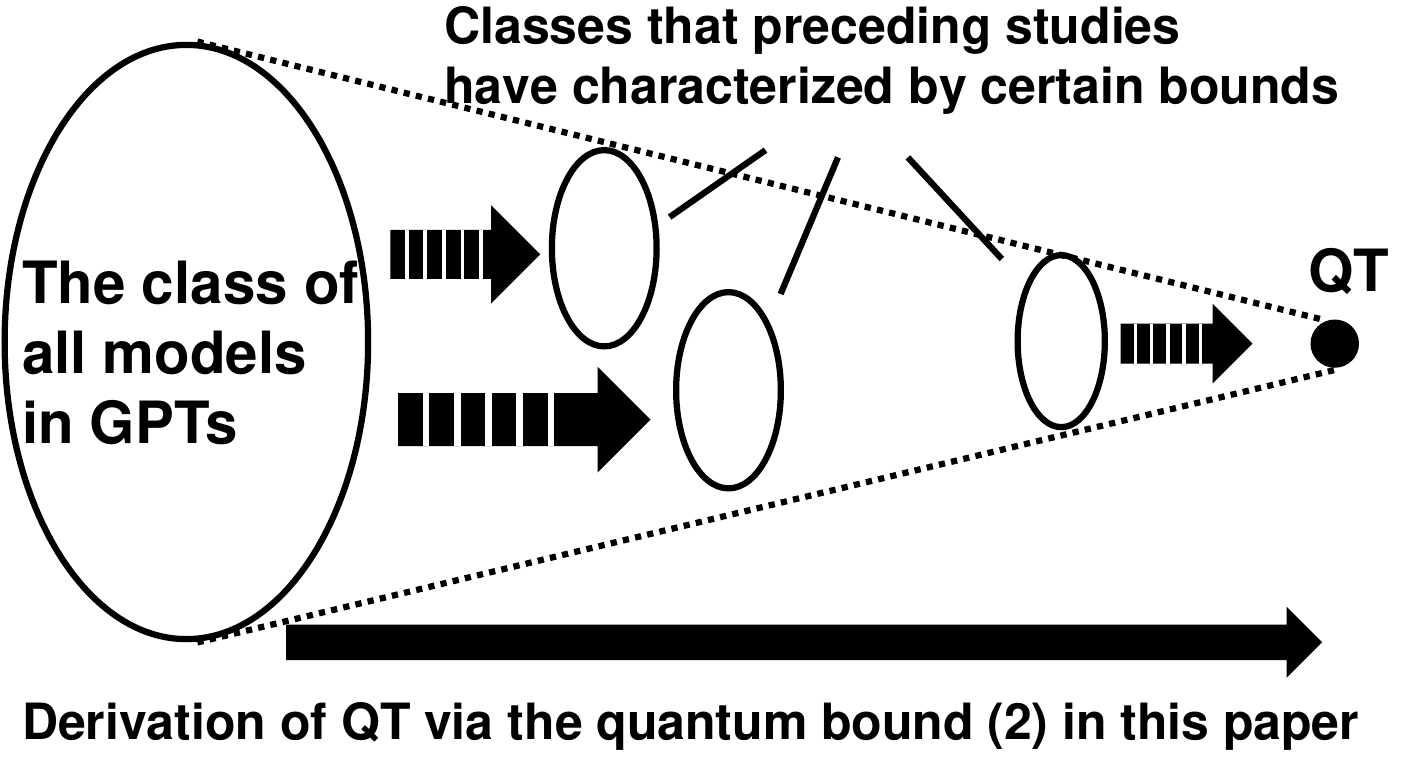}
	\caption{
	Preceding studies have considered many bounds for the performance for many information tasks.
	Some of them have characterized certain classes of models satisfying the same bound as Quantum Theory (QT) \cite{PR1994,Pawlowski2009,Barnum.Steering:2013,Arai2019,Short2010,Barnum2012,NLS2022,PNL2023},
	and others have derived QT out of certain classes of models \cite{YAH2020}.
	In contrast to the incompleteness of characterization of QT in such studies,
	we give a complete characterization of QT via the quantum bound for state discrimination \eqref{eq:trace-norm}.
	}
	\label{figure_class}
\end{figure}

\section{The setting of GPTs}

As a preliminary, we define a positive cone $\cC$ and its dual cone $\cC^\ast$ in finite-dimensional real vector space $\cV$.
A convex set $\cC\subset\cV$ is called a positive cone
if the following three conditions hold:
(i) for any $x\in\cC$ and any $r\ge0$, $rx\in\cC$. (ii) $\cC$ has a non-empty interior. (iii) $\cC\cap(-\cC)=\{0\}$.
For a positive cone $\cC$, the dual cone $\cC^\ast$ is defined as
$\cC^\ast:=\{m\in\cV^\ast\mid m(x)\ge0 \ \forall x\in\cC\}$.
A dual cone is also a positive cone, and we call an element $u\in\cC^\ast$ ordered unit if there exists $r\ge0$ such that $ru-m\in\cC^\ast$ for any $m\in\cC^\ast$.

By using the above mathematical objects, a model of GPTs is defined as follows.
\begin{definition}[A Model of GPTs]\label{def:model}
	A model of GPTs is defined as a tuple $\bm{G}=(\cV,\cC,u)$,
	where $\cV$, $\cC$, and $u$ are a real-vector space, a positive cone, and an \textit{order unit} of $\cC^\ast$, respectively.
\end{definition}
For a model of GPTs $\bm{G}$,
the state space and the measurement space are defined as follows.
\begin{definition}[State Space of GPTs]
	Given a model of GPTs $\bm{G}=(\cV,\cC,u)$,
	the state space of $\bm{G}$ is defined as
	\begin{align}\label{def:state}
		\cS(\bm{G}):=\left\{\rho\in\cC\middle|u(\rho)=1\right\}.
	\end{align}
	Here, we call an element $\rho\in\cS(\bm{G})$ a state of $\bm{G}$.
\end{definition}

%\begin{definition}[Effect Space]
%	Given a model  of GPTs $\bm{G}=(\cV,\cC,u)$,
%	the effect space of $\bm{G}$ is defined as
%	\begin{align}\label{def:eff}
%		\cE(\bm{G}):=\left\{E\in\cV\middle|E,u-E\in\cC^\ast\right\}.
%	\end{align}
%	Here, we call an element $e\in\cE(\bm{G})$ an effect of $\bm{G}$.
%	Also, we say that an effect $e$ is proper if $E$ satisfies $u-e\in\cE(\bm{G})$.
%\end{definition}

\begin{definition}[Measurements of GPTs]\label{def:measurement}
	Given a model  of GPTs $\bm{G}=(\cV,\cC,u)$,
	we say that a family $\{M_i\}_{i\in I}$ is a measurement
	if $M_i\in\cC^\ast$ and $\sum_{i\in I} M_i=u$.
	Besides, the index $i$ is called an outcome of the measurement.
	Here, the set of all measurements is denoted as $\cM(\bm{G})$.
	%Especially, the set of measurements with $n$-number of outcomes is denoted as $\cM_n(\bm{G})$.
\end{definition}

In this setting,
the state space and the measurement space are always convex.
%We call an element \textit{pure} when the element is extremal.
Also, when a state $\rho\in\cS(\bm{G})$ is measured by a measurement $\{M_i\}\in\cM(\bm{G})$,
the probability $p_i$ to get an outcome $i$ is given as
\begin{align}
	p_i:=M_i(\rho)
\end{align}

Standard quantum theory is a typical example of a model of GPTs,
i.e.,
standard quantum theory is given as the model $\bm{QT}:=(\cL_{\mathrm{H}}(\cH), \cL^+_{\mathrm{H}}(\cH)),\Tr)$,
where $\cL_{\mathrm{H}}(\cH)$ and $\cL^+_{\mathrm{H}}(\cH)$ denote the set of Hermitian matrices on a finite-dimensional Hilbert space $\cH$ and the set of positive semi-definite matrices on $\cH$, respectively.
In this model,
the state space $\cS(\bm{QT})$ is equal to the set of density matrices.
Also, by considering the correspondence from $m \in(\cL^+_{\mathrm{H}}(\cH))^\ast$ to $M\in\cL^+_{\mathrm{H}}(\cH))$ as $m(\rho)=\Tr \rho M$,
the measurement space $\cM(\bm{QT})$ corresponds to the set of Positive-Operator-Valued Measures (POVMs).

%Next, we introduce the domain of a given measurement in a model.
%\begin{definition}[Domain of Measurement]
%	Let $\bm{G}=(\cV,\langle,\rangle,\cC,u)$ be a model of GPTs.
%	We define the domain of a measurement $\{M_i\}_{i\in I}\in\cM(\bm{G})$ as
%	\begin{align}\label{def:domain-1}
%		\cD(\{M_i\}_{i\in I}):=\bigcap_{i\in I} \{X\in\cV\mid\Tr XM_i\ge0\}.
%	\end{align}
%\end{definition}
%domainについて何か書く．

Next,
we consider isomorphic maps between two models in order to introduce quantum-like models.
Let $\bm{G}_1=(\cV_1,\cC_1,u_1)$ and $\bm{G}_2=(\cV_2,\cC_2,u_2)$ be models of GPTs with $\dim(\cV_1)=\dim(\cV_2)$.
A linear isomorphic map $f:\cV_1\to\cV_2$ is called an isomorphic map of GPTs
if $f(\cC_1)=\cC_2$ and $u_2\circ f =c u_1$ for a constant $c>0$.
If such an isomorphic map exists,
these two models are equivalent up to normalization.
Actually, for any state $\rho\in\cS(\bm{G}_2)$ and any measurement $\{M_i\}\in\cM(\bm{G}_2)$,
the state $\rho':=cf^{-1}(\rho)\in\cS(\bm{G}_1)$ and the measurement $\{M_i':=\frac{1}{c}M_i\circ f\}\in\cM(\bm{G}_1)$
satisfy
\begin{align}
	M_i'(\rho')=\frac{1}{c}M_i\circ f \left(cf^{-1}(\rho)\right)=M_i(\rho).
\end{align}
In other words, these two models possess the same probabilistic structure.

From now on,
we compare the performance of measurements in general models with that of POVMs.
In this paper, we choose the trace norm $\|\cdot\|_1$ as a measure of the performance for state discrimination.
In order to compare the performance of general measurements with the bound of POVMs by the trace norm \eqref{eq:trace-norm},
we need to consider any model where the trace norm $\|\cdot\|_1$ can be defined.
Therefore, we temporarily restrict our target model to a \textit{quantum-like model},
i.e., a model $\bm{G}=(\Her{\cH},\cC,\Tr)$ satisfying $\cC\subset \Her{\cH}$.
In quantum-like models, we can define the trace norm straightforwardly.
A typical example of quantum-like but not quantum models is $\mathrm{SEP}$ defined by the positive cone $\cC=\{\sum_i x_i\otimes y_i \mid x_i\in\Psd{\cH_A}, \ y_i\in\Psd{\cH_B}\}$.
This model is known as one of composite models of standard quantum systems in GPTs \cite{Barnum.Steering:2013,Arai2019,YAH2020,Janotta2014,ALP2019,ALP2021,PNL2023,PNL2023}.
Hereinafter, we focus on quantum-like models at once.
One can consider this restriction strong,
but the following lemma states that the restriction is only the restriction of the dimension.
\begin{lemma}\label{theorem:isometry}
	Let $\bm{G}=(\cV,\cC,u)$ be a model of GPTs satisfying $\dim(\cV)=d^2$ for an integer $d\ge1$.
	Then,
	there exists an isomorphic map $f:\cV\to\Her{\cH}$ of GPTs from $\bm{G}$ to a quantum-like model $\tilde{\bm{G}}(f):=(\Her{\cH},f(\cC),\Tr)$.
\end{lemma}
Lemma~\ref{theorem:isometry} is mathematically related to Gleason's-type theorems given in \cite{Busch2003,CFMR2004}.
However, the statement of Lemma~\ref{theorem:isometry} is different from those in \cite{Busch2003,CFMR2004},
and therefore, we give the proof of Lemma~\ref{theorem:isometry} in \ref{sect:A-1}.
%In this paper, we impose the condition $\dim(\cV)=d^2$ to a model, and we only consider quantum-like models.
Here, we remark that the class of quantum-like models includes many important models.
For example, non-unique models of quantum composite systems always satisfy the condition $\dim(\cV)=d^2$.
Also, any model satisfies $\dim(\cV)=d^2$ by considering the composition with an ancillary classical system.

\section{General bound for state discrimination in GPTs}

Now, we consider single-shot state discrimination in a quantum-like model $\bm{G}=(\Her{\cH},\cC,\Tr)$.
For convenience,
we denote the element $f\in(\Her{\cH})^\ast$ as the element $M\in\Her{\cH}$ satisfying $\Tr M x=f(x)$ for any $x\in\cV$.
For example, the ordered unit $\Tr$ is denoted as the identity matrix $I$ through the correspondence.
Therefore, a measurement in $\cM(\bm{G})$ is given as a family $\{M_i\}_{i\in I}$ such that $\Tr M_i\rho\ge0$ for any $\rho\in\cC$ and $\sum_{i\in I} M_i=I$.

State discrimination is formulated similarly to the standard quantum theory.
Given two hypotheses $\rho_0,\rho_1$ for an unknown state $\rho$ in $\cS(\bm{G})$,
a player is required to determine which hypothesis is true
by applying a one-shot measurement $\{M_0,M_1\}\in\cM(\bm{G})$ 
with two outcomes $0$ and $1$.
The player supports the null-hypothesis $\rho=\rho_0$ 
if the outcome $0$ is observed
and supports the alternative hypothesis $\rho=\rho_1$ if the outcome $1$ is observed.
This decision has two types of error probabilities $\Tr \rho_0M_1$ and $\Tr \rho_1M_0$.
Here, we assume that the unknown state $\rho$ is prepared as $\rho_0$ and $\rho_1$ with probability $p$ and $1-p$, respectively.
Therefore, the total error probability of this decision is given as \eqref{def:error}.
We aim to minimize the value \eqref{def:error}.

%Since these two error probabilities have the trade-off relation, 
%we often minimize the sum of these two types of error probabilities \eqref{def:error}.
%Also, the value $1- \mathrm{Err}(\rho_0;\rho_1;p;\bm{M})$ can be considered as a distance in the state space even in general models \cite{}.

In the following,
the state discrimination under the standard quantum theory
is referred to 
the standard quantum state discrimination.
%In contrast, 
%the state discrimination under a general model is referred to a general state discrimination.
Under the standard quantum state discrimination,
the right-hand side of \eqref{eq:trace-norm} gives 
the minimum for optimizing the measurement $\bm{M}$.
The main purpose of this paper is to clarify %a precise condition 
what kind of general model outperforms 
the bound \eqref{eq:trace-norm} of the standard quantum state discrimination.

The references \cite{Arai2019,YAH2020} clarify the existence of 
a beyond-quantum measurement $\bm{M}$ that discriminates non-orthogonal pure states $\rho_0,\rho_1$ perfectly.
As mentioned in Introduction,
such examples mean the violation of the bound \eqref{eq:trace-norm}, i.e., the bound for the standard quantum state discrimination.
However,
these references do not answer the question of when the bound \eqref{eq:trace-norm} is violated.
%This paper answers the question as seen in the following two theorems.

To answer this question, we introduce a powerful tool for 
a two-outcome measurement $\bm{M}=\{M_0,M_1\}\in\cM(\bm{G})$.
We define the difference between the maximum and minimum eigenvalues
of $M_0$ as
\begin{align}
	r(\bm{M}):=\lambda_{\mathrm{max}}(M_0)-\lambda_{\mathrm{min}}(M_0),
\end{align}
where $\lambda_{\mathrm{max}}(M_0)$ and $\lambda_{\mathrm{min}}(M_0)$ are the maximum and minimum eigenvalues of $M_0$, respectively.
Because $M_1+M_2=I$,
the value $r(\bm{M})$ does not change even when $M_0$ is replaced by $M_1$.
Also, we define the sum of maximum and minimum eigenvalues
of $M_i$ as
\begin{align}
	r'(\bm{M},i):=\lambda_{\mathrm{max}}(M_i)+\lambda_{\mathrm{min}}(M_i).
\end{align}
The value $r'(\bm{M},i)$ depnds on the choise $i$, and the relation $r'(\bm{M},0)+r'(\bm{M},1)=2$ holds.

First, using the values $r(\bm{M})$ and $r'(\bm{M},i)$,
we give a general bound for the error probability of state discrimination as Theorem~\ref{theorem:performance}.
%in the case of $p=1/2$.
In the following, we denote 
the positive or negative part of $(\rho_0-\rho_1)$ by
$(\rho_0-\rho_1)_{+}$ and $(\rho_0-\rho_1)_{-}$, respectively.

\begin{theorem}\label{theorem:performance}
	Let $\bm{G}=(\Her{\cH},\cC,\Tr)$ be a quantum-like model.
	Any pair of two states $\rho_0,\rho_1\in\cS(\bm{G})$ and any measurement $\bm{M}=\{M_0,M_1\}\in\cM(\bm{G})$ satisfy
	\begin{align}\label{eq:performance}
		\mathrm{Err}(\rho_0;\rho_1;p;\bm{M})\ge \frac{1}{2}-\frac{1}{2}\left\|p\rho_0-(1-p)\rho_1\right\|_1r(\bm{M})
-\frac{1}{2}(2p-1)\left(r'(\bm{M},0)-1\right).
	\end{align}
The equality of \eqref{eq:performance} holds if and only if the following condition holds.
\begin{description}
\item[(A)]
Any normalized vector $\ket{\psi_{-}}$ in 
the range of $(p\rho_0-(1-p)\rho_1)_{-}$
belongs to the eigenspace of $M_0$ with the maximum eigenvalue,
and any normalized vector $\ket{\psi_{+}}$ in 
the range of $(p\rho_0-(1-p)\rho_1)_{+}$
belongs to the eigenspace of $M_0$ with the minimum eigenvalue.
\end{description}
\end{theorem}
The proof of Theorem~\ref{theorem:performance} is written in \ref{sect:A-2}.
Theorem~\ref{theorem:performance} reproducts the bound of POVMs \eqref{eq:trace-norm}
because an optimal POVM $\bm{M}$ satisfies $\lambda_{\mathrm{max}}(M_i)=1$ and $\lambda_{\mathrm{min}}(M_i)=0$, i.e., $r(\bm{M})=r'(\bm{M},i)=1$.
Especially in the case of $p=1/2$, the relation between the inequalities \eqref{eq:trace-norm} and \eqref{eq:performance} is more clear because the inequality \eqref{eq:performance}  is written as
\begin{align}
	\mathrm{Err}(\rho_0;\rho_1;\frac{1}{2};\bm{M})\ge \frac{1}{2}-\frac{1}{2}\left\|\frac{1}{2}\rho_0-\frac{1}{2}\rho_1\right\|_1r(\bm{M})
\end{align}
and any POVM $\bm{M}$ satisfies $r(\bm{M})\le 1$.
In a model of GPTs,
the value $r(\bm{M})$ can be larger than 1
because measurement effect $M_i$ can possess negative eigenvalues.
Therefore, in a model of GPTs,
the performance of state discrimination can be improved over the standard quantum.

Also, we remark that the inequality \eqref{eq:performance} is tight.
The following example of the tuple $\rho_0,\rho_1,p,\bm{M}$ satisfies the equality condition.
%\begin{example}\label{ex:1}
\textit{Example in separable cone}---%
We focus on the separable cone $\mathrm{SEP}_{2\times2}$ defined as
	\begin{align}
		\mathrm{SEP}_{2\times2}:=\left\{\rho=\sum_i \sigma_i\otimes\sigma_i'\in\Psd{\mathbb{C}^4}\middle| \sigma_i,\sigma_i'\in\Psd{\mathbb{C}^2}\right\},
	\end{align}
which is given as the set of unnormalized separable states in the $2\times2$-dimensional quantum system.
In the following, under the model $\bm{SEP}=(\Her{\mathbb{C}^4},\Tr,\mathrm{SEP}_{2\times2},I)$, we construct states $\rho_0,\rho_1$, probability $p=\frac{1}{2}$, and a measurement $\bm{M}$ 
that satisfy the following properties:
	\begin{align}
	\mathrm{Err}(\rho_0,\rho_1,p,\bm{M})=\frac{3}{8},
~\|p\rho_0-(1-p)\rho_1\|_1=1/8,
~r(\bm{M})=2,
\label{NTV}
	\end{align}
which imply the equality of \eqref{eq:performance}.

First, we choose the following matrices $\rho_0,\rho_1$:
	\begin{align}
		\rho_0:=\frac{1}{8}
		\begin{bmatrix}
			2&0&0&0\\
			0&2&1&0\\
			0&1&2&0\\
			0&0&0&2
		\end{bmatrix},\quad
		\rho_1:=\frac{1}{4}I.
	\end{align}
	The states $\rho_0$ and $\rho_1$ are separable
	because $\rho_0$ can be written as $\rho_0=\sum_{i=1}^7 \ketbra{\psi_i}{\psi_i}$,
	where
	\begin{align}
		&\ket{\psi_1}:=\frac{1}{4}(1,1)\otimes(1,1),\quad
		\ket{\psi_2}:=\frac{1}{4}(1,0)\otimes(1,-1),\nonumber\\
		&\ket{\psi_3}:=\frac{1}{4}(0,1)\otimes(1,-1),\quad
		\ket{\psi_4}:=\frac{1}{4\sqrt{2}}(1,-1)\otimes(1,i),\nonumber\\
		&\ket{\psi_5}:=\frac{1}{4\sqrt{2}}(1,-1)\otimes(1,-i),\quad
		\ket{\psi_6}:=\frac{1}{4\sqrt{2}}(1,i)\otimes(1,i),\nonumber\\
		&\ket{\psi_7}:=\frac{1}{4\sqrt{2}}(1,-i)\otimes(1,-i).
	\end{align}
	Also, the two matrices $\rho_0,\rho_1$ satisfy $\Tr \rho_i=1$ ($i=0,1$).
	Therefore, the two matrices $\rho_0,\rho_1$ belong to the state space $\cS(\bm{SEP})$.
Next, we choose the following family of matrices $\bm{M}=\{M_0,M_1\}$ as:
	\begin{align}
		M_0:=&
		\begin{bmatrix}
			1&0&0&0\\
			0&0&1&0\\
			0&1&0&0\\
			0&0&0&1
		\end{bmatrix}, \quad
		M_1:=&
		\begin{bmatrix}
			0&0&0&0\\
			0&1&-1&0\\
			0&-1&1&0\\
			0&0&0&0
		\end{bmatrix}.
	\end{align}
	First, $M_0+M_1=I$ holds.
	Next,
	because $M_1$ satisfies Positive Partial Transpose (PPT) condition,
	$M_0\in\mathrm{SEP}_{2\times2}^\ast$ \cite{Arai2019}.
	Also, $M_1\in\Psd{\mathbb{C}^4}\subset\mathrm{SEP}_{2\times2}^\ast$,
	and therefore, the family $\bm{M}$ belongs to the measurement space $\cM(\bm{SEP})$.

	Then, the tuple $\rho_0,\rho_1,p,\bm{M}$ satisfies the conditions \eqref{NTV} and the equivalent condition (A) for equality of \eqref{eq:performance}.
	They are easy to check, but we give a detailed calculation on \ref{sect:A-3} for reader's convenience.
%\end{example}

Theorem~\ref{theorem:performance} states that
the condition $r(\bm{M})>1$ is necessary for 
outperforming the standard state discrimination
in the case of $p=1/2$.
Moreover, given a measurement $\bm{M}$,
this condition is also sufficient for the existence of a pair of states
whose discrimination 
can be improved by the measurement $\bm{M}$
over the standard quantum state discrimination.
In other words, an equivalent condition for supremacy 
%of state discrimination 
over the standard quantum state discrimination
is simply given as the condition $r(\bm{M})>1$
in the case of $p=1/2$.
\begin{theorem}\label{theorem:equivalent}
	Let $\bm{G}=(\Her{\cH},\cC,\Tr)$ be a quantum-like model.
	Given a measurement $\bm{M}=\{M_0,M_1\}\in\cM(\bm{G})$,
	the following two conditions are equivalent:
	\begin{enumerate}
		\item There exist two states $\rho_0$ and $\rho_1$ in $\cS(\bm{G})$ such that 
		\begin{align}\label{eq:supremacy}
			\mathrm{Err}(\rho_0;\rho_1;\frac{1}{2};\bm{M}) < \frac{1}{2}-\frac{1}{2}\|\frac{1}{2}\rho_0-\frac{1}{2}\rho_1\|_1.
		\end{align}
		\item $r(\bm{M})> 1$.
	\end{enumerate}
\end{theorem}
The proof of Theorem~\ref{theorem:equivalent} is also written in \ref{sect:A-4}.
As shown in Theorem~\ref{theorem:equivalent},
a measurement $\bm{M}$ 
outperforms the standard quantum state discrimination 
if and only if the range $r(\bm{M})$ is strictly larger than 1.
As a typical example of such a superior measurement,
a measurement given in \cite{Arai2019} 
distinguishes two non-orthogonal separable states perfectly.
This example satisfies the statement 1 in Theorem~\ref{theorem:equivalent} with $p=\frac{1}{2}$
because 
perfect distinguishability corresponds to the equation $\mathrm{Err}(\rho_0;\rho_1;p;\bm{M}) =0$ and non-orthogonality corresponds to inequality $\|\frac{1}{2}\rho_0-\frac{1}{2}\rho_1\|_1<1$.
In this way,
the statement 2 is a simple equivalent condition for outperforming 
the standard quantum state discrimination.

Here, we simply remark the relation between
the violation of the bound \ref{eq:trace-norm} and the set of POVMs, i.e., the set $\cM(\bm{QT})$.
Even if the measurement $\bm{M}$ does not belong to $\cM(\bm{QT})$,
a measurement $\bm{M}$ does not always violate the bound \ref{eq:trace-norm}.
In other words,
there exists a measurement $\bm{M}$ satisfying $r(\bm{M})\le 1$ and $\bm{M}\not\in\cM(\bm{QT})$.
Theorem~\ref{theorem:equivalent} also ensures that such a measurement never improves the performance for state discrimination in terms of the value $\mathrm{Err}(\rho_0;\rho_1;p;\bm{M})$.

\section{Derivation of quantum theory and quantum simulability}

Now, we go back to the first aim,
i.e., the characterization of standard quantum theory by the bound of the performance for an information task.
Therefore, we do not restrict a model of GPTs.
By applying Theorem~\ref{theorem:equivalent},
we characterize the models isomorphic to standard quantum theory by the quantum bound
without the restriction to quantum-like models.
\begin{theorem}\label{theorem:derivation}
	Let $\bm{G}=(\cV,\cC,u)$ be a model of GPTs.
	The following conditions are equivalent:
	\begin{enumerate}
		\item There exists an isomorphic map $f:\cV\to\Her{\cH}$ from $\bm{G}$ to the model of standard quantum theory $\bm{QT}$.
		\item There exists an isomorphic map $f:\cV\to\Her{\cH}$ from $\bm{G}$ to a quantum like model $\tilde{\bm{G}}(f):=(\Her{\cH},f(\cC),\Tr)$ satisfying the following conditions (A) and (B).
		\begin{itemize}
			\item[(A).] The relation $\cS(\tilde{\bm{G}}(f))\subset\cS(\bm{QT})$ holds. 
			\item[(B).] Any two states $\rho_0,\rho_1\in\cS(\tilde{\bm{G}}(f))$, $0<p<1$, and any measurement $\bm{M}\in\cM(\tilde{\bm{G}}(f))$ satisfy the quantum bound \eqref{eq:trace-norm}.
		%\begin{align}\label{eq:derivation}
		%	\mathrm{Err}(\rho_0;\rho_1;p;\bm{M})\ge\frac{1}{2}-\frac{1}{2}\|p\rho_0-(1-p)\rho_1\|_1.
		%\end{align}
		\end{itemize}
	\end{enumerate}
\end{theorem}
The proof of Theorem~\ref{theorem:derivation} is written in \ref{sect:A-5}.
Here, we remark that there always exists an isomorphic map $f$ satisfying $\cS(\tilde{\bm{G}}(f))\subset\cS(\bm{QT})$ if $\dim(\cV)=d^2$
because $f(\cC)$ is spaned by $d^2-1$ linearly independent elements and we can take $d^2-1$ linearly independent elements in $\Psd{\cH}$.

Theorem~\ref{theorem:derivation} implies that
the state space of a model non-isomorphic to the model of standard quantum theory cannot be isometry-embedded in the standard quantum state space
with satisfying the quantum bound \eqref{eq:trace-norm}.
This statement operationally means that
a beyond-quantum model always outperforms state discrimination.
In contrast,
once we can embed the state space of a model in the standard quantum state space
with satisfying the quantum bound \eqref{eq:trace-norm},
the model must be a model of standard quantum theory
even though the measurement space is restricted only by the performance for state discrimination.
Especially, in the sence of embedding, the quantum bound \eqref{eq:trace-norm} implies all other properties in standard quantum theory.

\begin{remark}
Here, we emphasize that Theorem~\ref{theorem:derivation} is not trivial statement even with an isomorphic embedding of state space.
When we consider a isomorphic map $f:\cV\to\Her{\cH}$ from $\bm{G}$ to a quantum like model $\tilde{\bm{G}}(f):=(\Her{\cH},f(\cC),\Tr)$ satisfying the following conditions (A) in Theorem~\ref{theorem:derivation},
there must exist a measurement $\bm{M}\in\cM(\tilde{\bm{G}}(f))$ that does not belong to $\cM(\mathrm{QT})$.
However, as seen in Theorem~\ref{theorem:equivalent},
the measurement $\bm{M}$ violates the quantum bound \eqref{eq:trace-norm} in the case of $p=1/2$ if and only if $r(\bm{M})>1$,
which does not always hold even if a measurement does not belong to $\cM(\mathrm{QT})$.
Moreover, we can consider a sequence of models $\{\bm{G}_i\}_{i\in\mathbb{N}}$ whose state space $\cS(\bm{G}_i)$ converges $\cS(\bm{QT})$ with satisfying $\cS(\bm{G}_i)\subsetneq\cS(\bm{QT})$.
Theorem~\ref{theorem:derivation} ensures that there exists a measurement $\bm{M}\in\cM(\bm{G}_i)$ satisfying $r(\bm{M})>1$ for any $i$.
This is a non-trivial statement.

Here, we also remark that Thereom~\ref{theorem:derivation} is not directly shown by preceding studies about the norm determined by the error probability of state discrimination \cite{KMI2009,KNI2010-dist,NKM2010}.
In the preceding studies \cite{KMI2009,KNI2010-dist,NKM2010},
it is clarified that
the value $D(\rho_0,\rho_1)$ defined as follows for two states $\rho_i\in\cS(\bm{G})$ is a norm on any model of GPTs $\bm{G}$.
\begin{align}
	D_{\bm{G}}(\rho,\sigma):=\max_{\{e,u-e\}\in\cM(\bm{G})} \left(e(\rho_0)-e(\rho_1)\right).
\end{align}
Besides, the preceding studies \cite{KMI2009,KNI2010-dist,NKM2010} also showed the following error bound in a model $\bm{G}$.
\begin{align}\label{eq:norm}
	\mathrm{Err}(\rho,\sigma;\frac{1}{2};\bm{M})
	\ge \frac{1}{2}-\frac{1}{2}D_{\bm{G}}(\rho_0,\rho_1).
\end{align}
By applying \eqref{eq:norm},
the following relation holds in the embedding model $\bm{G}$ satisfying $\cS(\bm{G})\subset\cS(\bm{QT})$:
\begin{align}\label{eq:bound-norm}
	D_{\bm{G}}(\rho_0,\rho_1)\ge D_{\bm{QT}}(\rho_0,\rho_1)=\frac{1}{2}\|\rho_0-\rho_1\|.
\end{align}
However, the relation \eqref{eq:bound-norm} does not directly show Theorem~\ref{theorem:derivation} because the relation \eqref{eq:bound-norm} does not ensure that the equality \eqref{eq:bound-norm} never holds for a mode $\bm{G}$ except for $\bm{QT}$.
Theorem~\ref{theorem:derivation} rather shows the equality condition of \eqref{eq:bound-norm} as the following corollary.
\end{remark}

\begin{corollary}\label{cor:norm}
	For any quantum-like model $\bm{G}$ satisfying $\cS(\bm{G})\subset\cS(\bm{QT})$,
	the following conditions are equivalent:
	\begin{enumerate}
		\item $\bm{G}=\bm{QT}$.
		\item $D_{\bm{G}}(\rho_0,\rho_1)=\frac{1}{2}\|\rho_0-\rho_1\|$.
	\end{enumerate}
\end{corollary}
Corollary~\ref{cor:norm} is easily shown by Therem~\ref{theorem:derivation-1} in Appendix.

\begin{comment}
Theorem~\ref{theorem:derivation} assumes that the model $\bm{G}$ is quantum-like,
i.e., the dimension of its vector space is a square number.
This is the only assumption of our result.
We can relax the dimensional assumption for the aim to determine the composition of GPTs
because any model of composition of quantum systems is defined by a square dimensional vector space.
As a result, we obtain the following simple derivation of the composition of quantum theory without any assumption.
\begin{corollary}\label{corollary:derivation}
	Let $\bm{G}$ be a model of composition of quantum systems.
	The following conditions are equivalent:
	\begin{enumerate}
		\item There exists an isomorphic map $f$ from $\bm{G}$ to $\bm{QT}$.
		\item There exists an isomorphic map $f$ from $\bm{G}$ to $\tilde{\bm{G}}(f)$ satisfying $\cS(\tilde{\bm{G}}(f))\subset\cS(\bm{QT})$ and the quantum bound \eqref{eq:trace-norm}.
	\end{enumerate}
\end{corollary}
Corollary~\ref{corollary:derivation} is a new derivation of standard quantum composition from other compositions via the performance for state discrimination without any assumption.
In this sense,
Corollary~\ref{corollary:derivation} is a generalization of the result \cite{YAH2020},
which derives standard quantum composition from the restricted class of other compositions via the performance for state discrimination.
\end{comment}

\section{Discussion}

This paper has dealt with imperfect state discrimination in models of GPTs,
and we have compared the performance for state discrimination
between general measurements in quantum-like models and POVMs.
%In order to compare them with POVMs by trace norm,
%we have needed to consider a restricted class of models.
%However, we have clarified that any model with a condition of its dimension essentially satisfies the restriction
%(Lemma~\ref{theorem:isometry}).
%Next, we have discussed state discrimination
%in the restricted class.
We have introduced the range and the sum of eigenvalues of two-outcome measurement,
and we have given a general tight bound of the error sum by the range and the sum
(Theorem~\ref{theorem:performance}).
Besides, we have given an equivalent condition when a general measurement possesses superior performance to POVMs in the case of $p=1/2$ (Theorem~\ref{theorem:equivalent}).
As an application of the results of the performance for state discrimination in quantum-like models,
we have given a kind of the derivation of standard quantum theory out of all models of GPTs not restricted to quantum-like models (Theorem~\ref{theorem:derivation}).

Becafhe quantum bound \eqref{eq:trace-norm} derives the model of quantum theory,
the performance for state discrimination completely characterizes any other properties in standard quantum thery,
which is a surprising operational meaning of our derivation.
There exist many measures of the performance for information tasks that outperform the performance under the standard quantum theory in certain models of GPTs \cite{PR1994,Barnum.Steering:2013,Arai2019,YAH2020,NLS2022,NLS2022,PNL2023}.
However, all such known results do not completely characterize the model of quantum theory,
i.e.,
such measures sometimes behave in the same way as standard quantum theory.
Our finding is that the quantum bound of the performance for state discrimination is a completely characterizing measure.
In other words,
Theorem~\ref{theorem:derivation} operationally means that
the performance for state discrimination completely characterizes any other properties in standard quantum thery.

Finally, we give two important future directions for this work.
In this paper, we have given a derivation of standard quantum theory
as the existence of a state embedding isomorphic map satisfying the quantum bound,
an isomorphic map satisfying the conditions A and B in Theorem~\ref{theorem:derivation}.
Even if we do not restrict state embedding isomorphic maps,
no known example other than standard quantum theory
satisfies the quantum bound .
Therefore, it can be expected to remove the restriction for isomorphic maps, although such a relaxation is desired from the viewpoint of the quantum foundation.
If the relaxed statement is valid,
the bound of the performance for state discrimination can be regarded more strongly as an important physical principle.
The relaxation of the restriction for isomorphic maps
is the first important future direction of this work.

This paper has addressed only single-shot state discrimination.
In standard quantum information theory,
it is more important to clarify asymptotic behaviors of the $n$-shot case than the single-shot case.
The $n$-shot case is based on a $n$-composite system, which is difficult to deal with in GPTs because of the non-uniqueness of composite systems \cite{Janotta2014,Barrett2007}. 
Recently, the paper \cite{RLW2022} calculated the asymptotic performance for hypothesis testing with post-selection even in GPTs by applying
a result in the single shot case in general models.
However,
standard settings of hypothesis testing in GPTs, for example, 
the setting of Stein's lemma in GPTs, are still open.
This paper also has given a result in the single-shot case in general models as a general bound in Theorem~\ref{theorem:performance}.
The general bound in Theorem~\ref{theorem:performance} 
is applicable to the calculation of the performance for $n$-shot hypothesis testing, even in GPTs.
However, this application is still open because it requires various additional
calculations.
Therefore, the above required analysis
is the second important future direction of this work.

\section*{Acknowledgement}
HA appreciate Francesco Buscemi, Gen Kimura, Yui Kuramochi, Ryo Takakura, Shintaro Minagawa, and Kenji Nakanishi for giving helpful comments for the update of our results.
HA is supported by a JSPS Grant-in-Aids for JSPS Research Fellows No. JP22J14947.
MH is supported in part by the
National Natural Science Foundation of China (Grant No. 62171212).

\appendix

\section{}

\subsection{Proof of Lemma~\ref{theorem:isometry}}\label{sect:A-1}

\begin{proof}
Let $\bm{G}=(\cV,\cC,u)$ be a model of GPTs satisfying $\dim{\cV}=d^2$.
Take a basis $\{x_i\}_{i=1}^{d^2}$ in $\cV$ such that
$u(x_i)=1$ for $i\neq d^2$.
Also, take a basis $\{y_i\}_{i=1}^{d^2}$ in $\Her{\cH}$ such that
$\Tr y_i=1$ for $i\neq d^2$ and $\Tr y_{d^2}=u(x_{d^2})$.
Then, we choose an linear isomorphic map $f:\cV\to\Her{\cH}$ and a constant $c$ as the transformation from the basis $\{x_i\}_{i=1}^{d^2}$ to $\{y_i\}_{i=1}^{d^2}$ and $c=1$, respectively.
Now, we need to show two things: (i) $\Tr \circ f=u$ and (ii) $f(\cC)$ is a positive cone.
By the choice of $f$,
$u(x_i)=\Tr y_i$ holds for any $i$, which shows the statement (i).
Also, because $f$ is linear, $\cC$ is convex, and $\cC$ has non-empty interior,
the set $f(\cC)$ is also convex and has non-empty interior.
Besides, the relation $f(\cC)\cap f(-\cC)=f(\cC\cap -\cC)=f(\{0\})=\{0\}$ holds because $f$ is linear.
As a result, the statement (ii) holds,
and therefore, the model $\bm{G}=(\cV,\cC,u)$ is isomorphic to a quantum-like model $(\Her{\cH},f(\cC),\Tr)$.
\end{proof}

\subsection{Proof of Theorem~\ref{theorem:performance}}\label{sect:A-2}

\begin{proof}
We need to prove three statements;
(i) the inequality \eqref{eq:performance}, (ii) the implication ``(A) $\Rightarrow$ the equality of \eqref{eq:performance}", and (iii) the opposite implication ``the equality of \eqref{eq:performance} $\Rightarrow$ (A)".

\textbf{[Proof of (i)]}
Let $\bm{M}$ be a measurement with $r(\bm{M})\le1$.
Here, we denote the positive part and the negative part of a Hermitian matrix $x$ as $x_+$ and $x_-$, respectively.
Any two Hermitian matrices $x_1,x_2$ with $\Tr x_1=\Tr x_2=1$
satisfy
\begin{align}
	\Tr (px_1-(1-p)x_2)_+ +\Tr(px_1-(1-p)x_2)_-=\Tr (px_1-(1-p)x_2)=2p-1.
\end{align}
%Therefore, $-\Tr(x_1-x_2)_-=\Tr(x_1-x_2)_+$ holds.

Then, the following calculation shows the inequality \eqref{eq:performance},
\begin{align}\label{eq:proof-performance-1}
	&\mathrm{Err}(\rho_0;\rho_1;p;\bm{M})=p\Tr \rho_0 M_1+(1-p)\Tr \rho_1 M_0 \nonumber\\
	=&p\Tr \rho_0(I-M_0)+(1-p)\Tr \rho_1 M_0 \stackrel{(a)}{=} p -\Tr (p\rho_0-
	(1-p)\rho_1)M_0\nonumber\\
	=&p-\frac{1}{2}\Tr(p\rho_0-(1-p)\rho_1)2M_0\nonumber\\
	=&p-\frac{1}{2}\left(\Tr(p\rho_0-(1-p)\rho_1)M_0+\Tr(p\rho_0-(1-p)\rho_1)(I-M_1)\right)\nonumber\\
	\stackrel{(b)}{=}&\frac{1}{2}-\frac{1}{2}\left(\Tr(p\rho_0-(1-p)\rho_1)M_0-\Tr(p\rho_0-(1-p)\rho_1)M_1\right)\nonumber\\
	\stackrel{(c)}{\ge}&\frac{1}{2}-\frac{1}{2}\Bigl(
	\lambda_{\mathrm{max}}(M_0)\Tr(p\rho_0-(1-p)\rho_1)_+
	+\lambda_{\mathrm{min}}(M_0)\Tr(p\rho_0-(1-p)\rho_1)_- \nonumber\\
	&-\lambda_{\mathrm{min}}(M_1)\Tr(p\rho_0-(1-p)\rho_1)_+
	-\lambda_{\mathrm{max}}(M_1)\Tr(p\rho_0-(1-p)\rho_1)_-\Bigr)\nonumber\\
	\stackrel{(d)}{=}&\frac{1}{2}-\frac{1}{2}\Bigl(
	\lambda_{\mathrm{max}}(M_0)\Tr(p\rho_0-(1-p)\rho_1)_+
	-\lambda_{\mathrm{min}}(M_0)\left(\Tr(p\rho_0-(1-p)\rho_1)_+ -(2p-1)\right)\nonumber\\
	&+\lambda_{\mathrm{min}}(M_1)\left(\Tr(p\rho_0-(1-p)\rho_1)_- -(2p-1)\right)
	-\lambda_{\mathrm{max}}(M_1)\Tr(p\rho_0-(1-p)\rho_1)_- \Bigr)\nonumber\\
	=&\frac{1}{2}-\frac{1}{2}\Bigl(
	\Tr(p\rho_0-(1-p)\rho_1)_+r(\bm{M})-\Tr(p\rho_0-(1-p)\rho_1)_- r(\bm{M})\Bigr)\nonumber\\
&-\frac{1}{2}(2p-1)\left(\lambda_{\mathrm{min}}(M_0)-\lambda_{\mathrm{min}}(M_1)\right)\nonumber\\
	\stackrel{(e)}{=}&\frac{1}{2}-\frac{1}{2}||p\rho_0-(1-p)\rho_1||_1r(\bm{M})-\frac{1}{2}(2p-1)(r'(\bm{M},0)-1).
\end{align}
The equations $(a)$ and $(b)$ are shown by $\Tr \rho_0=\Tr \rho_1=1$.
The inequality $(c)$ holds because of the inequality
\begin{align}
	&\Tr(p\rho_0-(1-p)\rho_1)M_i\nonumber\\
	\le&\lambda_{\mathrm{max}}(M_i)\Tr(p\rho_0-(1-p)\rho_1)_+
	+\lambda_{\mathrm{min}}(M_i)\Tr(p\rho_0-(1-p)\rho_1)_-
\end{align}
holds for any $i=0,1$.
The equation $(d)$ is shown by the equation $2p-1-\Tr(x_1-x_2)_-=\Tr(x_1-x_2)_+$.
The equation $(e)$ holds because the equation $\|x\|_1=\Tr (x_+-x_-)$ holds for any Hermitian matrix $x$
and the relation $\lambda_{\mathrm{min}}(M_1)=1-\lambda_{\mathrm{max}}(M_0)$ holds.

\textbf{[Proof of (ii)]}
We assume that the condition (A).
The vector $\ket{\psi_{-}}$ belongs to 
		the eigenspace of $M_0$ with the maximum eigenvalue,
and 		$\ket{\psi_{+}}$ belongs to 
		the eigenspace of $M_0$ with the minimum eigenvalue.
Then, we have
\begin{align}
	&\Tr(p\rho_0-(1-p)\rho_1)M_0\nonumber\\
	=&\lambda_{\mathrm{max}}(M_0)\Tr(p\rho_0-(1-p)\rho_1)_+
	+\lambda_{\mathrm{min}}(M_0)\Tr(p\rho_0-(1-p)\rho_1)_-.\label{NV1}
\end{align}
Due to the equation $M_0+M_1=I$,
the vector $\ket{\psi_{-}}$ belongs to 
		the eigenspace of $M_1$ with the minimum eigenvalue,
and 		$\ket{\psi_{+}}$ belongs to 
		the eigenspace of $M_1$ with the maximum eigenvalue.
Then, we have
\begin{align}
	&\Tr(p\rho_0-(1-p)\rho_1)M_1\nonumber\\
	=&\lambda_{\mathrm{min}}(M_1)\Tr(p\rho_0-(1-p)\rho_1)_+
	+\lambda_{\mathrm{max}}(M_1)\Tr(p\rho_0-(1-p)\rho_1)_-.\label{NV2}
\end{align}
The combination of 
\eqref{NV1} and \eqref{NV2}
implies the equality of the inequality $(c)$ in \eqref{eq:proof-performance-1}. As a result, the quality of \eqref{eq:performance} holds.

\textbf{[Proof of (iii)]}
We assume the equality of \eqref{eq:performance}.
Therefore,
the equality $(c)$ holds in \eqref{eq:proof-performance-1},
i.e., we obtain the following equation:
\begin{align}
	&\Tr(p\rho_0-(1-p)\rho_1)M_0-\Tr(p\rho_0-(1-p)\rho_1)M_1\nonumber\\
	=&\lambda_{\mathrm{max}}(M_0)\Tr(p\rho_0-(1-p)\rho_1)_+
	+\lambda_{\mathrm{min}}(M_0)\Tr(p\rho_0-(1-p)\rho_1)_- \nonumber\\
	&-\lambda_{\mathrm{min}}(M_1)\Tr(p\rho_0-(1-p)\rho_1)_+
	-\lambda_{\mathrm{max}}(M_1)\Tr(p\rho_0-(1-p)\rho_1)_-.
\end{align}
As seen in the above calculation \eqref{eq:proof-performance-1}, the left-hand side is equal to $2\Tr(\rho_0-\rho_1)M_0$.
Also, the requirement of measurement $M_0+M_1=I$ implies the two equations $\lambda_{\mathrm{min}}(M_1)=1-\lambda_{\mathrm{max}}(M_0)$ and $\lambda_{\mathrm{max}}(M_1)=1-\lambda_{\mathrm{min}}(M_0)$.
Therefore, the right-hand side is equal to
\begin{align}
	2\left(\lambda_{\mathrm{max}}(M_0)\Tr(p\rho_0-(1-p)\rho_1)_+
	+\lambda_{\mathrm{min}}(M_0)\Tr(p\rho_0-(1-p)\rho_1)_-\right).
\end{align}
As a result, we obtain the following equation:
\begin{align}
	0=&\Tr(p\rho_0-(1-p)\rho_1)M_0-\lambda_{\mathrm{max}}(M_0)\Tr(p\rho_0-(1-p)\rho_1)_+\nonumber\\
	&-\lambda_{\mathrm{min}}(M_0)\Tr(p\rho_0-(1-p)\rho_1)_-\nonumber\\
	=&\Tr(p\rho_0-(1-p)\rho_1)_+(-\lambda_{\mathrm{max}}(M_0)I+M_0)\nonumber\\
	&+\Tr(p\rho_0-(1-p)\rho_1)_-(-\lambda_{\mathrm{min}}(M_0)I+M_0).\label{eq:proof-equality-1}
\end{align}
Because four relations $(\rho_0-\rho_1)_+\ge0$, $(\rho_0-\rho_1)_-\le0$, $-\lambda_{\mathrm{max}}(M_0)I+M_0\le0$, and $-\lambda_{\mathrm{min}}(M_0)I+M_0\ge0$ hold,
the equation \eqref{eq:proof-equality-1} implies the following two equations:
\begin{align}
	\Tr(p\rho_0-(1-p)\rho_1)_+(-\lambda_{\mathrm{max}}(M_0)I+M_0)&=0,\\
	\Tr(p\rho_0-(1-p)\rho_1)_-(-\lambda_{\mathrm{min}}(M_0)I+M_0)&=0.
\end{align}
These equalities imply the condition (A).
\end{proof}

\subsection{Detailed Check for the Conditions on the Example}\label{sect:A-3}

Here, we check that the tuple $(\rho_0,\rho_1,p,\bm{M})$ in the section of example satisfies the conditions in \eqref{NTV} and the equivalent condition (A) for equality of \eqref{eq:performance}.

The condition $\mathrm{Err}(\rho_0,\rho_1,\bm{M})=3/8$ is easy to check by the definition $\mathrm{Err}(\rho_0,\rho_1,\bm{M}):=\Tr p\rho_0M_1+(1-p)\rho_1M_0$.

Since the matrix $\rho_0-\rho_1$ has the spectral decomposition
\begin{align}\label{eq-proof-example-1}
		\rho_0-\rho_1
		=\frac{1}{8}
		\begin{bmatrix}
			0&0&0&0\\
			0&1/2&-1/2&0\\
			0&-1/2&1/2&0\\
			0&0&0&0
		\end{bmatrix}
		-
		\frac{1}{8}
		\begin{bmatrix}
			0&0&0&0\\
			0&1/2&1/2&0\\
			0&1/2&1/2&0\\
			0&0&0&0
		\end{bmatrix},
\end{align}
the condition $\|p\rho_0-(1-p)\rho_1\|_1=\frac{1}{2}\|\rho_0-\rho_1\|=1/8$ is obtained.

Since the matrix $M_1$ satisfies $\rank(M_1)=1$ and $\Tr M_1=2$,
the relations $\lambda_{\mathrm{max}}(M_1)=2$ and $\lambda_{\mathrm{min}}(M_1)=0$ hold.
Therefore, the condition $r(\bm{M})=2$ holds.
As a result, we complete the check of conditions \eqref{NTV}.

Finally, we check the equivalent condition (A) for equality of \eqref{eq:performance}.
As seen in the spectral decomposition \eqref{eq-proof-example-1}
of $(p\rho_0-(1-p)\rho_1)$,
we find that 
the vector $\ket{\psi_+}$ is 
$(0,\frac{1}{\sqrt{2}},-\frac{1}{\sqrt{2}},0)^T$,
and
the vector $\ket{\psi_-}$ is 
$(0,\frac{1}{\sqrt{2}},\frac{1}{\sqrt{2}},0)^T$.
The eivenvector of $M_0$ with the eigenvalue $-1$ is
$(0,\frac{1}{\sqrt{2}},-\frac{1}{\sqrt{2}},0)^T$.
The eigenspace of $M_0$ with the eigenvalue $1$ is
the orthogonal space to 
$(0,\frac{1}{\sqrt{2}},-\frac{1}{\sqrt{2}},0)^T$.
Since 
the vector $\ket{\psi_-}$ belongs to 
the eigenspace of $M_0$ with the maximum eigenvalue
and 
the vector $\ket{\psi_+}$ is the 
the eigenvector of $M_0$ with the minimum eigenvalue,
we obtain the condition (A).

\subsection{Proof of Theorem~\ref{theorem:equivalent}}\label{sect:A-4}

\begin{proof}
The statement 1$\Rightarrow$2 is implied by Theorem~\ref{theorem:performance}.
We will show the statement 2$\Rightarrow$1 in the case of $p=\frac{1}{2}$.

Let $\bm{M}=\{M_0,M_1\}$ be an arbitrary measurement in $\bm{G}$ satisfying the condition 2.
The Hermitian matrix $M_0$ can be written as $M_0=\sum_{i=1}^d \lambda_i E_i$ by the spectral decomposition in the descending order of the eigenvalnes.
The matrix $M_1$ is also written as $M_1=I-M_0=\sum_{i=1}^d (1-\lambda_i)E_i$.
Due to the condition 2,
the inequality $r(\bm{M})=\lambda_d-\lambda_1>1$ holds.

Next, as a preliminary,
we will choose a number $\epsilon>0$ satisfying a certain property as follows.
Because $\cC$ is a positive cone,
there exists an inner point $x_0\in\cC$.
In other words,
there exists a number $\epsilon>0$ such that the $\epsilon$-neighborhood $N_\epsilon(x_0):=\{x\in\cV\mid ||x_0-x||_2\le \epsilon \}$ is contained by $\cC$.
Therefore, the set $S_\epsilon(x_0):=\{\rho\in N_\epsilon(x_0)\mid \Tr \rho=1\}$ is contained by the state space $\cS(\bm{G})$.

Next,
in order to choose two states $\rho_0$ and $\rho_1$,
we take two elements $\delta_1$ and $\delta_2$ such that
\begin{align}
\begin{aligned}
	\|x_0 +\delta_1 E_1 -\delta_1 E_d\|_2&\le\epsilon,\\
	\|x_0 -\delta_2 E_1 +\delta_2 E_d\|_2&\le\epsilon.
\end{aligned}
\end{align}
Then, 
by taking $\delta_0:=\min\{\delta_1,\delta_2\}$,
the two elements $x_0 +\delta_0 E_1 -\delta_0 E_d$ and $x_0 -\delta_0 E_1 +\delta_0 E_d$ belong to $N_\epsilon(x_0)$.
Therefore, the following two elements belong to $S_\epsilon(x_0)$:
\begin{align}
\begin{aligned}
	\rho_0:&=\frac{1}{\Tr x_0} \left( x_0+\delta_0 E_1 -\delta_0 E_d\right),\\
	\rho_1:&=\frac{1}{\Tr x_0} \left(x_0 -\delta_0 E_1 +\delta_0 E_d\right).
\end{aligned}
\end{align}

Finally, the following calculation shows that two states $\rho_0$ and $\rho_1$ satisfy the desirable inequality \eqref{eq:supremacy}.
\begin{align}
	&\mathrm{Err}(\rho_0;\rho_1;p;\bm{M})=\frac{1}{2}\Tr \rho_0 M_1+\frac{1}{2}\Tr \rho_1 M_0\nonumber\\
	=&\frac{1}{2}\Tr \rho_0 M_1 +\frac{1}{2}\Tr \rho_1(I-M_1)
	=\frac{1}{2}-\frac{1}{2}\Tr (\rho_1-\rho_0)M_1\nonumber\\
	=&\frac{1}{2} -\Tr \delta_0(E_d-E_1)M_1
	=\frac{1}{2}-\Tr \delta_0 (\lambda_d-\lambda_1)\nonumber\\
	\stackrel{(a)}{=}&\frac{1}{2}-\frac{1}{4}||\rho_0-\rho_1||_1(\lambda_d-\lambda_1)
	\stackrel{(b)}{<}1-\frac{1}{2}||\frac{1}{2}\rho_0-\frac{1}{2}\rho_1||_1.
\end{align}
The equation $(a)$ is shown by $||\rho_0-\rho_1||_1=4\delta_0$.
The inequality $(b)$ holds because $\lambda_d-\lambda_1>1$.
As a result, we complete the proof of 2$\Rightarrow$1.
\end{proof}

\subsection{Proof of Theorem~\ref{theorem:derivation}}\label{sect:A-5}

For the proof of Theorem~\ref{theorem:derivation},
we give the following theorem.

\begin{theorem}\label{theorem:derivation-1}
	Let $\bm{G}=(\Her{\cH},\Tr,\cC,I)$ be a quantum-like model of GPTs.
	Under the condition $\cS(\bm{G})\subset\cS(\bm{QT})$,
	the following conditions are equivalent:
	\begin{enumerate}
		\item $\cC=\Psd{\cH}$, i.e., $\bm{G}=\bm{QT}$.
		\item Any two states $\rho_0,\rho_1\in\cS(\bm{G})$ and any measurement $\bm{M}\in\cM(\bm{G})$ satisfy
		\begin{align}\label{eq:derivation}
			\mathrm{Err}(\rho_0;\rho_1;\frac{1}{2};\bm{M})\ge\frac{1}{2}-\frac{1}{2}\|\frac{1}{2}\rho_0-\frac{1}{2}\rho_1\|_1.
		\end{align}
	\end{enumerate}
\end{theorem}

\begin{proof}[Proof of Theorem~\ref{theorem:derivation-1}]
	The statement 1$\Rightarrow$2 holds because the inequality \eqref{eq:derivation} is shown in standard quantum information theory (for example, \cite{Holevo1972,HelstromBook1976,HayashiBook2017}).
	We will show the statement 2$\Rightarrow$1 by contraposition.
	
	Let $\bm{G}=(\Her{\cH},\Tr,\cC,I)$ be a model of GPTs satisfying
	$\cS(\bm{G})\subset\cS(\bm{QT})$.
	The condition $\cS(\bm{G})\subset\cS(\bm{QT})$ implies $\cC\subset\Psd{\cH}$,
	and therefore, we obtain
	$\cC^\ast\supset\Psd{\cH}$ because the relation $\cC_1\subset\cC_2$ is equivalent to the relation $\cC_1^\ast\supset\cC_2^\ast$ for any two positive cones $\cC_1,\cC_2$.
	To show the contraposition of the implication 2$\Rightarrow$1,
	we assume that $\cC\neq\Psd{\cH}$ and equivalently $\cC^\ast\supsetneq\Psd{\cH}$.
	Therefore, there exists an element $M\in\cC^\ast\setminus\{0\}$ such that $M\not\in\Psd{\cH}$.
	In other words, the inequality $\lambda_{\mathrm{min}}(M)<0$ holds.
	Because $\Tr \rho M\ge0$ for any $\rho\in\cC\subset\Psd{\cH}$,
	the inequality $\lambda_{\mathrm{max}}(M)>0$ holds.
	Then, a Hermitian matrix $M'$ defined as
	$M':=M/\lambda_{\mathrm{max}}(M)$ satisfies
	\begin{align}
		\lambda_{\mathrm{max}}(M')=1, \quad \lambda_{\mathrm{max}}(M')<0,
	\end{align}
	which implies $r(\bm{M})>1$.
	Also, the matrix $I-M'$ satisfies
	\begin{align}
		\lambda_{\mathrm{max}}(I-M')=0, \quad \lambda_{\mathrm{max}}(I-M')>1,
	\end{align}
	which implies $I-M'\in\Psd{\cH}$.
	Therefore, the family $\{M',I-M'\}$ belongs to $\cM(\bm{G})$.
	Hence, Theorem~\ref{theorem:equivalent} ensures that there exists two states $\rho_0,\rho_1\in\cS(\bm{G})$ satisfying the inequality \eqref{eq:supremacy} for the measurement $\{M',I-M'\}$,
	which implies that condition 2 in Theorem~\ref{theorem:derivation} does not hold.
	As a result, we complete the proof of statement 1$\Rightarrow$2 by contraposition.
\end{proof}

By applying Theorem~\ref{theorem:derivation-1},
we prove Theorem~\ref{theorem:derivation} as follows.

\begin{proof}[Proof of Theorem~\ref{theorem:derivation}]
	If  $\dim(\cV)\neq d^2$, both of the two conditions in Theorem~\ref{theorem:derivation} are false,
	i.e., they are equivalent.
	Therefore, we need to prove the statement in the case of $\dim(\cV)=d^2$.
	
	First, we prove the implication $1\Rightarrow2$.
	Because of the condition $1$,
	there exists an isomorphic map $f:\cV\to\Her{\cH}$ from $\bm{G}$ to $\bm{QT}$.
	By the map $f$,
	the transformed model $\tilde{\bm{G}}(f)$ is $\bm{QT}$,
	and therefore, the condition $2$ holds as the quantum bound.
	
	Second,
	we prove the implication $2\Rightarrow1$ by the contraposition.
	Therefore, we assume that there does not exists an isomorphic map from $\bm{G}$ to $\bm{QT}$.
	We need to show there does not exists an isomorphic map $f$ satisfying both of the conditions A. and B.
	Without loss of generality,
	we consider an arbitrary isomorphic map $f$ satisfying the condition A., and we need to show that the map $f$ never satisfies the condition B.
	Because $f$ is an isomorphic map satisfying $\cS(\tilde{\bm{G}}(f))\subset\cS(\bm{QT})$ and $\tilde{\bm{G}}(f)$ is not $\bm{QT}$,
	Theorem~\ref{theorem:derivation-1} ensures that there exists a tuple of states, a probability, and a measurement breaking the quantum bound.
	In other words, the condition B. does not hold.
	As a result,
	the implication $2\Rightarrow1$ has been proven,
	and the proof is finished.
\end{proof}

\end{document}